\documentclass{amsart}
\usepackage{amsmath, amssymb, verbatim}

\usepackage[all]{xy}
\newcommand{\h}{\mathcal{H}}
\newcommand{\g}{\mathfrak{g}}
\newcommand{\R}{\mathbb{R}}
\newcommand{\C}{\mathbb{C}}

\newcommand{\A}{\mathcal{A}}

\newcommand{\be}{\begin{eqnarray}}
\newcommand{\bes}{\begin{eqnarray*}}
\newcommand{\ee}{\end{eqnarray}}
\newcommand{\ees}{\end{eqnarray*}}
\newcommand{\Spec}{\textrm{Spec }}

\newcommand{\Hom}{\textrm{Hom}}

\newtheorem{theorem}{Theorem}[section]

\newtheorem{lem}[theorem]{Lemma}
\newtheorem{prop}[theorem]{Proposition}

\theoremstyle{definition}
\newtheorem{definition}{Definition}
\newtheorem{example}{Example}

\theoremstyle{remark}
\newtheorem{remark}{Remark}

\newcommand{\id}{\operatorname{id}}

\begin{document}

\title{A perspective on regularization and curvature}
\author{Susama Agarwala}
\address{California Institute of Technology\\1200 E California Ave.\\ Pasadena, CA 91107}

\maketitle

\begin{abstract} 
A global connection on the Connes Marcolli renormalization bundle
relates $\beta$-functions of a class of regularization schemes by
gauge transformations, as well as local solutions to $\beta$-functions
over curved space-time.
\\ \textbf{Mathematical Subject Classification
  (2010): 81R99} 

\end{abstract}

%\tableofcontents

\section{Introduction}

The process of regularization and renormalization is well known to
physicists studying Quantum Field Theories (QFTs) and can be found in
textbooks such as \cite{Ti} chapters 18-21, and
\cite{IASW}. Regularization is the process of rewriting an undefined
quantity in terms of certain parameters such that the quantity is well
defined away from a predetermined limit of the parameters, and
renormalization makes sense of the regularized quantity at the
limit. There are many ways of regularizing a QFT, and little is known
about how to relate different methods. Furthermore, solving a QFT for
physical values after regularization is a local process. There is no
known way to solve a QFT over a curved space-time background (i.e. a
QFT that has been coupled to gravity). This paper shows a geometric
way of conceptualizing the relationship between regularization
schemes, and of local solutions to regularized QFTs.

Many different types of regularization processes commonly
used, such as dimensional regularization, Pauli-Villars
regularization, momentum cut-off regularization, $\zeta$-function
regularization, and point splitting regularization. The choice of
regularization scheme depends on the symmetries of the QFT under
study, and the ease of calculation of the scheme, among other
factors. These schemes are not equivalent to each other and very
little is known about relationships between them.

Renormalization extracts a well defined physical value at a
predetermined limit of the regularization parameters introduced, that
matches the value observed in experiment. A common method of
renormalization of Feynman amplitudes is BPHZ renormalization, which
is an algorithm for iteratively subtracting off terms in a regularized
QFT that would lead to divergences. The key object necessary to solve
for numerical values from a regularized Lagrangian is called a
$\beta$-function. Analytically, the contribution to the
$\beta$-function of for any regularization scheme has not been solved
for graphs containing more than a few loops. Different regularization
schemes give rise to different values of the
$\beta$-function. Furthermore, efforts at studying the
$\beta$-function over a curved space-time background (i.e. for a QFT
coupled with gravity) have only been successful locally.

In 2001 \cite{CK01}, Connes and Kreimer address the first of these
drawbacks by first showing that BPHZ renormalization for dimensional
regularization of a scalar field theory is exactly the process of
Birkhoff decomposition of loops. They then express the contributions
of arbitrary graphs to the $\beta$-function in terms of the Birkhoff
decomposed loops.  In 2006, \cite{CMbook} Connes and Marcolli rewrite
this process in terms of a renormalization bundle over the space
parameterized by the regulator and identify a class of connections on
this bundle defined uniquely by the $\beta$-function.

In this paper I notice that this setup allows one to study
renormalization schemes beyond dimensional regularization. I identify
a global connection on the renormalization bundle. Sections of the
renormalization bundle differ from each other either by the
regularization scheme they represent, or the parameters of the
Lagrangian of the QFT they represent. The choice of regularization
scheme and Lagrangian define a gauge field on the bundle. Identifying
a global connection relates pullbacks of the connection along sections
via gauge transformations. In the first case, the global connection
relates different regularization schemes, including regularization
schemes that do not have well defined $\beta$-functions, i.e.
non-renormalizable regularization schemes, and ones that do, i.e.
renormalizable ones.  The second case gives a way of understanding
renormalization over curved space time. Currently, $\beta$-function
calculations are done in coordinate patches, and it is difficult to
check for consistency of results across these patches. Computations on
two different coordinate patches correspond to different sections of
the renormalization bundle. The existence of a global connection
implies the existence of a global $\beta$-function over the general
manifold.

Section two of this paper reviews the development of the tools
necessary for the construction of the renormalization bundle,
following \cite{CK00}, \cite{CMbook}, \cite{EM}. Section 3 discusses the
physical and geometrical $\beta$-function, following \cite{IASW} and
\cite{CMbook}. Section 4 defines the global section on the
renormalization bundle.

\section{The Connes Marcolli renormalization bundle}

In this paper, I work with the renormalizable scalar quantum field
theory of valence 3 interactions. It can be defined by the Lagrangians
of the form \be \mathcal{L}= \frac{1}{2} (|d\phi|^2 -m^2\phi^2) +
g\phi^3 \; , \label{StLag} \ee where $m$ is the mass of the parameter,
and $g$ is the coupling constant. I use this particular Lagrangian to
stay consistent with the work in \cite{CK00}, \cite{CK01} and
\cite{CMbook}. 
These interactions can be depicted graphically in Feynman diagrams.

\subsection{Hopf algebra}

Feynman graphs can be given a Hopf algebra structure by considering the one particle reducible, or 1PI, graphs that make up the general graphs. 

\begin{definition}
A 1PI graph is a connected Feynman graph such that the removal of any
internal edge still results in a connected graph.
\end{definition}

A Hopf algebra can be built out of the Feynman diagrams by assigning
variables $x_\Gamma$ to each 1PI graph $\Gamma$ and considering the
polynomial algebra on these variables $\h = \C[\{x_\Gamma| \Gamma
  \textrm{ is 1PI }\}]$. This Hopf algebra is constructed in
\cite{CK00}. The product of two variables in this algebra
$m(x_{\Gamma_1}\otimes x_{\Gamma_2}) = x_{\Gamma_1}x_{\Gamma_2}$
corresponds to the disjoint union of graphs, and the unit is given by
the empty graph, $1_\h = x_\emptyset$. 

To construct the co-product, I need to review the definition of
admissible subgraphs.

\begin{definition}
Let $V(\Gamma)$ be the set of vertices of a graph $\Gamma$,
$I(\Gamma)$ the set of internal edges and $E(\Gamma)$, the set of
external edges. The Feynman diagram $\gamma$ is an admissible subgraph
of a 1PI Feynman diagram $\Gamma$ if and only if the following
conditions hold:
\begin{enumerate}
\item The Feynman diagram $\gamma$ is a 1PI Feynman diagram,
  or a disjoint union of such diagrams.
\item Let $\gamma'=\gamma \setminus E(\gamma)$ be the diagram $\gamma$
  without its external edges. There is an embedding $i: \gamma'
  \hookrightarrow \Gamma$ that preserves the field type of each edge.
\item The set of edges (internal and external) meeting the vertex $v
  \in V(\gamma)$ is he same as the set of edges meeting $i(v) \in
  V(\Gamma)$.
\end{enumerate}
\end{definition}

The last condition ensures that the external leg conditions are
preserved under the embedding. Finally, I recall a definition of a
contracted graph to represent the divergences that remain after the
subtraction of the subdivergences.

\begin{definition}Let $\gamma$ be a disconnected admissible
subgraph of $\Gamma$ consisting of the connected components $\gamma_1
\ldots \gamma_n$. A contracted graph $\Gamma//\gamma$, is the Feynman
graph derived by replacing each connected component $i(\gamma'_j)$,
with a vertex $v_{\gamma_j} \in V(\Gamma//\gamma)$. \end{definition}

The coproduct of this Hopf algebra is given by the subgraph and
contracted graph structure of the Feynman diagrams \bes \Delta
x_\Gamma = 1 \otimes \Gamma + \Gamma\otimes 1 + \sum_{\gamma
  \subsetneq \Gamma}x_\gamma \otimes x_{\Gamma//\gamma}\ees where the
sum is taken over all proper admissible subgraphs of $\Gamma$. The
unit is a map \bes \eta: \C &\rightarrow& \h \\ 1 &\mapsto& 1_\h \ees
and the co-unit
can be defined on generators of $\h$ as \bes \varepsilon : \h
&\rightarrow& \C \\ 1_\h &\mapsto& 1 \\ x_{\Gamma \neq \emptyset}
&\mapsto& 0 \;.\ees The kernel of the co-unit is the ideal generated
by all $x_\Gamma$ such that $\Gamma$ is non-empty. The antipode is
defined to satisfy the antipode condition for Hopf
algebras \begin{eqnarray*}S:\h &\rightarrow& \h \\x_\Gamma
  &\rightarrow& -x_\Gamma -\sum_{\gamma \subset \Gamma}m(S(x_\gamma)
  \otimes x_{\Gamma \slash \slash \gamma})\;. \end{eqnarray*} This is
a bigraded Hopf algebra, with one grading given by loop number and the
other by insertion number. Let $\h^n$ be the $n^{th}$ graded element of
$\h$ by loop number. If $x_\Gamma \in \h^n$, then $\Gamma$ has $n$
loops. The grading operator $Y$ on $\h$ is defined on generators as
\bes Y: \h^n &\rightarrow& \h^n \\ x_\Gamma &\mapsto& n x_\Gamma \;
.\ees Details on the two grading structures are given in \cite{CK00}
and \cite{BK06}. This Hopf algebra is associative, co-associative
\cite{CK00} and commutative, but not co-commutative.

In general, Hopf algebras can be interpreted as a ring of functions on
a group. Since the spectrum of a commutative ring is an affine space,
the group in question is affine group scheme, $G = \Spec \h$. The
group laws on the Lie group $G$ are covariantly defined by the Hopf
algebra properties \bes (\id\otimes \Delta) \Delta = (\Delta \otimes
\id)\Delta & \leftrightarrow & \rm{multiplication} \\ (\id \otimes
\varepsilon)\Delta = \id & \leftrightarrow & \rm{identity} \\ m(S
\otimes \id) \Delta = \varepsilon \eta & \leftrightarrow &
\rm{inverse} \ees The group $G$ can also be viewed as a functor from a
$\C$ algebra $A$ to $G(A) = \Hom_{\textrm{alg}}(\h, A)$. The affine
group scheme $G$ is developed in detail in \cite{CMbook}. The last
condition above means that if $\gamma \in G(A)$, and $x\in \h$, then
$\gamma^{-1}(x) = S(\gamma(x)) = \gamma(S(x))$.

The Lie algebra $\g$ associated to $G$ is the infinitesimal elements
$\delta_\gamma$, where \bes \delta_{\gamma_1}(x_{\gamma_2})
= \begin{cases} 1 & \text{$\gamma_1= \gamma_2$,} \\ 0 &
  \text{else.} \end{cases} \ees By the Milnor-Moore theorem, the
universal enveloping algebra is isomorphic to the restricted dual of
$\h$ \bes \mathcal{U}(\g) \simeq \h^\vee = \oplus_n\h^{n*}\; \ees
where the grading is given by the loop number of the graph. The
restricted dual is the direct sum of the duals of each graded
component of $\h$. The product is defined on $\h^\vee$ by the
convolution product \bes \alpha_1\star \alpha_2 (x_\Gamma) =
m(\alpha_1 \otimes \alpha_2)(\Delta x_\Gamma) \quad \alpha_i \in
\h^\vee \;.\ees This is described in detail in \cite{CK00} and
\cite{Men}. The convolution product on $\g$ acts as an insertion
operator on $\h$. For two generators of $\h$, $x_{\Gamma_1}$ and
$x_{\Gamma_2}$, define \bes x_{\Gamma_1} \star x_{\Gamma_2} =
\sum_{x_\Gamma} m(\delta_{\Gamma_1} \otimes \delta_{\Gamma_2})(\Delta
x_\Gamma) \cdot x_\Gamma \ees where the sum is taken over all
generators of $\h$. This product induces an insertion product on the
1PI graphs of a theory in the same fashion that the coproduct on $\h$
is induced by the subgraph structure on the 1PI graphs. This
convolution product induces a pre Lie structure on the generators on
the 1PI graphs of a theory. The Lie bracket \bes
[x_{\Gamma_1},x_{\Gamma_2}] = x_{\Gamma_1} \star x_{\Gamma_2} -
x_{\Gamma_2}\star x_{\Gamma_1} \ees follows the Jacobi identity, as
can be checked. For details on this construction, see \cite{CK00} and
\cite{EK06}. The grading operator $Y$ can be defined on $\h^\vee$ as
$Y(\gamma(x)) = \gamma(Y(x))$.

Manchon \cite{Man} develops bijective correspondence between $G(A)$
and a $\g(A)$ defined as \bes \tilde{R} : G(A) &\rightarrow &\g(A)
\\ \gamma &\mapsto& \gamma^{\star -1} \star Y(\gamma) \; .\ees Manchon
also shows that this is inverse of the time ordered expansional defined
by Connes and Marcolli in \cite{CMbook} \bes Te : \g(A) &\rightarrow&
G(A) \\ \alpha &\mapsto & Te^{\int_a^b\theta_{-s}(\alpha) ds}\ees

\begin{remark}
The time ordered expansional is not the same bijection as that of the
standard exponential map from $\g$ to $G$, \bes \rm{exp(\alpha)} =
\epsilon + \alpha + \frac{\alpha^2}{2!} + \ldots \ees The time ordered
expansional is given by the formula \bes Te^{\int_a^b\theta_{-s}
  (\alpha) ds}= \epsilon + \sum_{n=1}^\infty\underbrace{Y^{-1}(\ldots
  Y^{-1}(}_{n \textrm{ times}}\alpha)\ldots) \;. \ees In fact the
operator $S \star Y (\gamma)$ is closely related to the Dynkin
operator on the commutative Hopf algebra $\h$ of Feynman graphs
\cite{BEP}.
\end{remark}

\subsection{Birkhoff decomposition}

In \cite{CK00}, Connes and Kreimer show that BPHZ renormalization can
be written as a composition of loops in the Lie group $G$ using the
Birkhoff decomposition theorem. The following is a summary of their
results.

Let $\A = \C\{\{z\}\}$ be the algebra of formal Laurent series in $z$
with poles of finite order. Then $\Spec \A = \Delta^*$, the punctured
infinitesimal disk around the origin in $\C$. Let $\gamma(z)$ be a map
from a simple loop not containing the origin in $\Delta^*$ to
$G$. There is a natural isomorphism from the group of these maps and
$G(\A)$.  By the Birkhoff decomposition theorem, $\gamma(z)$
decomposes as the product \bes \gamma(z) = \gamma_-^{-1}(z)\star
\gamma_+(z) \;,\ees where $\gamma_+(z)$ is a well defined map in the
interior of the loop (containing $z=0$), and $\gamma_-^{-1}(z)$ is a
well defined map outside of the loop (away from $z=0$). Each
$\gamma(z)$ can be written as a Laurent series with poles of finite
order and coefficients in $G(\C)$ convergent in $\Delta^*$. 
%The
%algebra homomorphisms decompose as $\gamma(z) = \gamma_-^{\star
%  -1}(z)\star \gamma_+(z)$, where $\star$ is the product on $\h^\vee$
%and $G(A)$. 
The map $\gamma_+(z)$ is a somewhere convergent formal power series in
$z$, and for $x_\Gamma \not \in \ker(\varepsilon)$,
$\gamma_-(z)(x_\Gamma) = \sum_{-n}^{-1} a_i z^i$, where $a_i \in
G(\C)$. Finally, normalizing $\gamma_-(z)(x_\emptyset) = 1_\h$,
following \cite{CK00}, ensures the uniqueness of th Birkhoff
decomposition. 

Connes and Kreimer \cite{CK00} show that the recursive formula for
calculating $\gamma_+(z)(x_\Gamma)$ and $\gamma_-(z)(x_\Gamma)$ is the
exact same as the recursive formula for calculating the renormalized
and counterterm contributions respectively of a Feynman diagram
$\Gamma$ to the regularized Lagrangian given by BPHZ. For $\Gamma$ a
1PI graph, $\gamma(z)(x_\Gamma)$ is the value of the regulated Feynman
integral of the graph $\Gamma$, $\lim_{z\rightarrow
  0}\gamma_+(z)(x_\Gamma)$ is the renormalized value of the graph
while $\gamma_-(z)(x_\Gamma)$ is the counterterm. 

\begin{remark} 
Elements of $G(\A)$ correspond to QFTs regulated by a complex
parameter. Any $\phi^3$ scalar field theory in $6$ dimensions under
any regularization scheme that yields results in $\C\{\{z\}\}$
corresponds to a $\gamma(z) \in G(\A)$. Connes and Kreimer's work in
\cite{CK00} extends to a huge class of Lagrangians and regularization
schemes. \label{sections}
\end{remark}

\section{The $\beta$-function}
The regularization process results in a Lagrangian that is a function
of the regularization parameter. Prior to regularization, the
Lagrangian of any theory is scale invariant. That is \bes \int_{\R^n}
\mathcal{L}(x) \, d^nx = \int_{\R^n} \mathcal{L}(x)\, d^n(tx)\; .\ees
When the Lagrangian is regularized, and written in terms of a
regularization parameter, $z$, it is no longer scale
invariant. Specifically, the counterterms of a theory depends on the
scale of the Lagrangian. In order to preserve scale invariance in the
regularized Lagrangian one introduces a regularization mass, which is
also a function of the regularization parameter, to cancel out any
scaling effects introduced by regularization.

\subsection{Derivation in physics}
The renormalization group describes how the dynamics of Lagrangian
depends on the scale at which it is probed. One expects that probing
at higher energy levels reveals more details about a system than at
lower energies. To go from higher energy to lower, average over the
extra information at the higher energy, $\lambda$, and rewrite it in
terms of a finite number of parameters at a lower energy, $\mu$. %The
%Lagrangian at the lower energy scale is called the \emph{effective
%  Lagrangian at $\mu$}, $(\mathcal{L},\mu)$. 
For a specified set of fields and interactions the Lagrangian,
$\mathcal{L}$ at an energy scale $\mu$ is written $(\mathcal{L}, \mu)$
where $\mathcal{L}$ has coefficients depending on $\mu$.

Formally, let $M \simeq \mathbb{R}_+$ be a non-canonical energy space,
with no preferred element. Fix a set of fields and interactions. Call
$S$ the set of Lagrangians for this system in the energy
space, $M$. For $\lambda, \mu \in M$ such that $\lambda > \mu$, there
is a map \be R_{\lambda, \mu}: \quad S \rightarrow
S \label{Raction}\ee so that the  Lagrangian at $\mu$ is
written $R_{\lambda, \mu}\mathcal{L}$ for $\mathcal{L} \in S$. The map
in \eqref{Raction} can be written as an action of $(0,1]$ on $S \times
  M$: \be \begin{array}{ccc}(0, 1] \times (S \times M) &\rightarrow& S
      \times M \\t \circ (\mathcal{L}, \lambda) &\mapsto& (R_{\lambda,
        t\lambda}\mathcal{L}, t\lambda) \;
      .\end{array} \label{semidirprodish}\ee In the
  Lagrangian $R_{\lambda, t\lambda}\mathcal{L}(t)$, all parameters,
  $m$, $\phi$, and $g$ are functions of the mass scale $t$. The map
  $R_{\lambda, \mu}$ satisfies the properties
\begin{enumerate}
\item $ R_{\lambda, \mu}R_{\mu,\rho} = R_{\lambda, \rho} \; .$
\item  $R_{\lambda, \lambda} = 1 \; .$
\end{enumerate}

\begin{definition} The set $\{R_{\lambda, \mu}\}$ forms a semi-group
called the renormalization group in the physics
literature.\label{renormgrp}\end{definition}

The renormalization group equations can be derived from
differentiating the action in \eqref{semidirprodish} and solving \be
\frac{\partial}{\partial t} (R_{\lambda, t\lambda}\mathcal{L}_{ct}) =
0 \label{local} \; .\ee This differential equation gives rise to a
system of differential equations that describe the $t$ dependence of
the unrenormalized parameters, $m(t)$, $g(t)$ and $\phi(t)$, in
$R_{\lambda, t\lambda}\mathcal{L}(t)$. To solve the renormalization
group equations, it is sufficient to solve for $g(t)$. \emph{The
  $\beta$-function describes the $t$ dependence of $g$} and can be
written as \bes\beta(g(t)) = t \frac{\partial g(t)}{\partial t}\;.\ees

This above development of the renormalization group and
renormalization group equations follows \cite{IASG}. For details on
the renormalization group equations for a $\phi^4$ theory, QED and
Yang-Mills theory, see \cite{Ti} chapter 21 or \cite{Ry} Chapter 9.

Connes and Marcolli show that the $\beta$-function of a
renormalization theory is an element of $\g$. The quantities listed
above are the sums of the $ \beta$ function evaluated on the one loop
graphs. That is, the geometric $\beta$-function for a section $\gamma$
is given by \bes \beta(\gamma) =
\sum_{x_\Gamma} \beta(\gamma)(x_\Gamma) \; ,\ees where the sum is taken
over the $x_\Gamma$ generating $\h^1$.

\subsection{As a geometric object}

The geometric $\beta$-function requires a more general construction of
the renormalization group and Lagrangians. In the
renormalization bundle, the non-canonical energy space is given by $M
\simeq \C^\times$. The space $S$ of Lagrangians is replaced
by the space $G(\A)$, the space of evaluators of a regularized
 Lagrangians. The renormalization group is a group in this
generalization (not just a semi group) given by $\theta_s = e^{sY}$
for $s \in \C$.  The action of the renormalization group can be
written as a $\C^\times$ action that factors through $\C$ by setting
$t(s) = e^s$ \bes t^Y : \quad G(\A) &\rightarrow & G(\A) \\ \gamma(z)
&\mapsto & t^Y\gamma(z) = \gamma_t(z) \; . \ees The space $S \times M$
becomes $\tilde{G}(\A) = G(\A) \rtimes_\theta \C^\times$ in the
notation of \cite{CMbook}. The action of $\C^\times$ on $\tilde{G}(\A)$
is given by \be \begin{array}{ccc} \C^\times \times \tilde{G}(\A)
  &\rightarrow& \tilde{G}(\A)  \\ t \circ (\gamma, \lambda)
  &\mapsto & (t^Y\gamma, t\lambda) \;
  . \end{array} \label{ctimesaction} \ee

The renormalization bundle, $P^* \rightarrow B^*$ is a $\tilde{G}(\A)$
principle bundle. By the $\C^\times$ action in \eqref{ctimesaction},
it is a $C^\times$ invariant bundle.  The base space $B^* \simeq
\Delta^* \times \C^\times$ is a product of the regularization
parameter and the non-canonical energy space. In this context, the
$\beta$-function is given by \bes \beta(\gamma(z)) =
\frac{d}{dt}|_{t=1} \lim_{z\rightarrow 0} \gamma(z)^{\star -1}\star
t^{zY}(\gamma(z)) = \lim_{z\rightarrow 0}z\tilde{R}(\gamma) \; . \ees
This is only well defined when $\gamma(z)$ satisfies condition
\eqref{local}. To find the derivation of the geometric
$\beta$-function in this context, see \cite{CK01}, \cite{CMbook} or
\cite{EM}.

\section{A global connection\label{con}}
This section develops a global connection on the Connes-Marcolli
renormalization bundle. A standard result from differential geometry
\cite{Spivak} shows that a connection defined on the base space of a
bundle can be written as the pullback along a section of a global
connection defined on the top space of the bundle.  The connections on
$B^*$ identified by Connes and Marcolli in \cite{CMbook} correspond to
the pullbacks of a single global connection on $P^*$.

Let $\omega$ be that connection on $P^*$, defined on pullbacks along
sections by the logarithmic differential operator, as in
\cite{CMbook}.
\begin{definition} Let $D$ be a differential operator.
\bes D:\tilde{G}(\A) &\rightarrow& \Omega^1(\tilde\g) \\
(\gamma (z),t) &\mapsto&
(\gamma (z),t)^{\star -1}\star d(\gamma(z),t)\;. \ees
\end{definition}

Many of the properties of the connection
in \cite{CMbook} extend to the global connection.

\begin{lem}For $f, \; g \in \tilde{G}(\A)$, the differential
$D(f) = f^\star \omega$ defines a connection on
  section $f$ of $P^* \rightarrow B^*$.
\label{logmult}\end{lem}

\begin{proof}
If $D$ defines a connection, it must satisfy equation \be (f^{\star
  -1} \star g)^*\omega = g^{-1}dg + g^{\star -1}(f^*\omega) g \;
,\label{pullbackcond}\ee for $f,\, g \in \tilde{G}(\A)$. Since
$df^{-1} =- f^{-1} df f^{-1}$, \bes D(f^{-1} g) = Dg - g^{-1}f
f^{-1}df f^{-1}g \;, \ees or \bes Dg = D(f^{-1} g) + (f^{-1}g)^{-1}Df
(f^{-1}g) \;.\ees which satisfies equation
\eqref{pullbackcond}. 
\end{proof}

\begin{prop} The connection $\omega$ is $\C^\times$ equivariant,
\bes u^Y\omega(z, t, x) = \omega (z, ut, u^Yx ) \; .\ees \end{prop}

\begin{proof}
The proof given in \cite{CMbook} of this statement for
$\gamma_\mu^\star\omega$ generalizes to all sections $(\gamma(z), t)$,
and thus to the entire connection. Since $P^*\rightarrow B^*$ is a
$\C^\times$ equivariant bundle.
\end{proof}

Since $(\gamma (z), t) = t\circ(t^{-Y}\gamma(z), 1)$,
by the $\C^\times$ action on $\tilde{G}(\A)$, and
$(t^Y\gamma (z), 1)$ is identified with $\gamma _t$, it
is sufficient to define the connection of sections of the form \bes
(t^Y \gamma (z), 1)^* \omega = \gamma_t^*\omega \; .\ees

\begin{prop}Given any section $\gamma_{t}$, one can directly calculate
the corresponding pullback of the connection $\omega$ on it.
\[D(t^Y\gamma (z)) = t^{Y}(\gamma^{ \star -1}(z) \star
\partial_z \gamma (z))dz + t^{Y}(\tilde{R} (\gamma ) (z)
)\frac {dt}{t}\;.\]\end{prop}

\begin{proof} One has \[d(t^{Y}\gamma (z)) =
t^{Y}(\partial_z\gamma (z))dz + t^{Y} \gamma  (z) \star
\tilde{R} (\gamma (z)) \frac{dt}{t}\; .\] Multiplying on the
left by $\gamma^{  \star -1}$ gives the logarithmic
derivative \[Dt^{Y} \gamma  (z) = t^{Y}(\gamma^{ \star
  -1}(z) \star \partial_z\gamma (z))dz + t^{Y}(\tilde{R}
(\gamma )(z) )\frac{dt}{t}\;.\]
\end{proof}

Since $\omega \in \Omega^1(\tilde{\g})$, $\omega$ has the form \bes
(\gamma, t)^*\omega = a_\gamma(z, t)dx + b_\gamma(z, t)\frac{dt}{t}
\\ (\gamma^*_t)\omega = a_{\gamma_t}(z, 1)dx + b_{\gamma_t}(z,
1)\frac{dt}{t} \;.\ees 
The terms $a_{\gamma_t}$ and $b_{\gamma_t}$ are defined as
\be a_{\gamma_t}(z,1) = \gamma^{ \star -1}(z) \star
\partial_z \gamma (z) \notag \\
 b_{\gamma_t}(z,1) = \tilde{R}(\gamma)(z)
\label{bdef} \;.\ee

\begin{prop}The connection $\omega$ is flat. \end{prop}
\begin{proof}It is sufficient to check that each pullback 
is flat. That is, that all the pullbacks satisfy \bes
[a_{\gamma_t}(z,1), b_{\gamma_{t}}(z,1)] =
\partial_t(a_{\gamma_{t}}(z,1)) -
\partial_z(b_{\gamma_{t}}(z,1))\;.\ees
\end{proof}

Given this explicit form of $a_{\gamma_t}$ and $b_{\gamma_t}$, I can
state the main theorem of this paper.

\begin{theorem}Let $\omega$ be a global connection on
the bundle $P^* \rightarrow B^*$ defined on sections of the bundle by
the differential equation $\gamma_{t}^*\omega = D \gamma_{t}
(z)$. Then $\omega$ is uniquely defined by $\tilde{R}:
G(\A)\rightarrow \g(\A)$.\label{mainthm}\end{theorem}

\begin{proof} 
Let $\omega$ be a connection on $P^* \rightarrow B^*$. Since $\omega$
is $\C^\times$ invariant, it is sufficient to consider pullbacks of
$\omega$ along sections $(\gamma_t, 1)$, which can be written \bes
\gamma_t^*\omega = a_{\gamma_t}(z, 1)dx + b_{\gamma_t}(z, 1)
\frac{dt}{t} \;.\ees If \bes a_{\gamma_t}(z, 1) = Te^{-\int_0^\infty
  \theta_{-s}b_{\gamma_t}(z, 1) ds}\star \partial_z
Te^{\int_0^\infty \theta_{-s}b_{\gamma_t}(z, 1) ds} \ees then \bes
\gamma_t^*\omega = D Te^{\int_0^\infty \theta_{-s}b_{\gamma_t}(z, 1) ds}
\; . \ees

Conversely, defining the connection by the pullback of its sections, we
see that $\gamma^*\omega$ is uniquely defined by $\gamma$. The section
$\gamma$ is uniquely defined by the map $\tilde{R}(\gamma ) \in
\g(\A)$. \end{proof}

Theorem \ref{mainthm} is a generalization of the main result of Connes
and Marcolli in \cite{CMbook}. Notice that if $\gamma(z) =
\gamma_-^{\star -1} \star \varepsilon(z)$, i.e. if $\gamma(z)$ is a
counterterm for some regularized Lagrangian, then \be \gamma_-^{\star
  -1}(z)^* \omega = D Te^{\int_0^\infty
  \theta_{-s}(\tilde{R}(\gamma_-^{\star -1})) ds} \;
. \label{counterterm}\ee If $\gamma(z)$ satisfies \eqref{local},
Ebrahimi-Fard and Manchon show that \be \lim_{z\rightarrow 0}
\tilde{R}(\gamma) = \beta(\gamma_{t})= \beta(\gamma_-^{ \star -1})
\label{bbeta}\; \ee and that this definition is equivalent to the 
Connes Kreimer definition for the $\gamma(z)$ corresponding to
dimensional regularization. Furthermore, if $\gamma(z)$ does not
satisfy \eqref{local}, then the $\lim_{z\rightarrow
  0}\tilde{R}(\gamma)$ is not well defined.  Therefore, if $\gamma$
satisfies \eqref{local} then \bes \gamma_-(z) = Te^{-\int_0^\infty
  \theta_{-s}(\frac{\beta(\gamma_-)}{z}) ds} \ees as in \cite{CMbook}.

The logarithmic differential operator defining the connection $\omega$
has a symmetry under right multiplication by sections of the form
$\gamma = \varepsilon \star \gamma_+(z)$. Let $\A_+= \C[[z]]$ be the
algebra of formal power series in $z$ with $\C$ coefficients. Notice
that the sections $\gamma_+(z)$ form the group $G(\A_+)$.

\begin{definition} 
Two pullbacks of the connection $\gamma_{t}^*\omega$ and
$\gamma_{t}^{\prime *}\omega$ are equivalent if and only if one
pullback can be written in terms of the action of $G(\A_+)$ on the other \bes
\gamma^{\prime *}_{t}\omega= D\psi_{t} + \psi_{t}^{\star -1} \star
\gamma_{t} ^*\omega \star \psi_{t}\ees for $\psi_{t} \in G(\A_+)_{t}$,
the group of sections that are regular in $z$ and $t$. I write this
equivalence as $\gamma_{t}^{\prime *}\omega \sim
\gamma_{t}^*\omega$.
\end{definition}

\begin{remark} 
This gauge equivalence is the same as the statement
$\gamma'_t = \gamma_t \star \psi_{t}$, and specifically, $\gamma'_{t-}
= \gamma_{t-}$. The gauge equivalence on the connection classifies
pullbacks by the counterterms of the corresponding sections.
\end{remark}

Connes and Marcolli's equisingular connection is the special case of
pulling back $\omega$ along a section, $\gamma$, for which
$\beta(\gamma)$ is well defined. Recall the definition

\begin{definition} The pullback  $\gamma^*_{t}\omega$ along
on $P^* \rightarrow B^*$ is equisingular when pulled back
to the bundle $P^* \rightarrow \Delta^*$ if and only if
\begin{itemize}\item $\omega$ is equisingular under the $\C^\times$
action on the section of $P^* \rightarrow B^*$.
\item For every pair of sections $\sigma, \; \sigma'$ of the
  $B\rightarrow \Delta^*$ bundle, $\sigma(0) = \sigma' (0)$, the
  corresponding pull backs of the connection $\omega$,
  $\sigma^*(\gamma_{t}^*\omega)$ and
  $\sigma^{\prime*}(\gamma_{t}^*\omega)$ are equivalent under the
  action of $G(\A_+)$. 
\end{itemize}
\end{definition}

The following summarizes the important properties of equisingular
connections.

\begin{prop}
The following statements are equivalent: \begin{enumerate} 
\item $\gamma_t^*\omega$ is an equisingular connection on $B^*$ 
\item Let $D\gamma(z, t) = \omega$. The counterterm is independent of
  the renormalization mass parameter \bes \frac{d}{dt}\gamma_-(z,t) = 0 \;.\ees This is the same as equation \eqref{local}. 
\item Write \bes \gamma_t^*\omega = a(z,1) dz + b(z, 1) \frac{dt}{t}
  \ees One can write \bes b(z,1)= \tilde{R}(\gamma_t) = \sum_{i=-1}^\infty \alpha_i(\gamma) z^i \ees
  where $\alpha_i(\gamma) \in \g(\C)$
\item The coefficient $\alpha_{-1} = \beta(\gamma)$ determines the
  pullback of the connection $\gamma^*\omega$ up to a $G(\A_+)$
  equivalence.
 \end{enumerate} \label{equivs}\end{prop}
\begin{proof} 
\begin{description}
\item{1 $\iff$ 4} That $\gamma^*\omega$ is determined by
  $\beta(\gamma)$ is proved in \cite{CMbook}. That $\beta(\gamma) =
  -a_{-1}(\gamma)$ is shown in \cite{EM} if 2 holds.
\item{1 $\iff$ 2} Rewriting $\sigma^*(\gamma_t) =
  \gamma_{\sigma(z)}$, one sees that $\frac{d}{dt}\sigma^*(\gamma_t) =
  \frac{d}{d\sigma}(\gamma_\sigma(z))$. Therefore,
  $\frac{d}{dt}(\gamma_t)_- = 0 \iff (\gamma_t)_-$ does not depend on
  $\sigma(z)$. This is exactly the second condition in the definition
  of equisingularity.
\item{2 $\Rightarrow$ 3} Proved in \cite{EM}, and by equation \eqref{bdef}
\item{3 $\Rightarrow$ 4} Follows from \eqref{bbeta} and \eqref{counterterm}.
\end{description}
\end{proof}

As stated in Remark \ref{sections}, choosing a regularization scheme
for a specified Lagrangian fixes a section of the renormalization
bundle. The global connection $\omega$ means that these sections
differ only by a gauge transform. From Proposition \ref{equivs}, if
two sections that the same counterterm which has a well defined
$\beta$-function, they are equivalent under gauge transformation.

\begin{example} 
A good example of this latter fact is $\zeta$-function
regularization. A propagator in $\zeta$-function regularization is
regularized \bes \int_{\R^6}\frac{1}{(p^2+m^2)^{1+s}} d^6p \ees while
under dimensional regularization it is \bes
\int_{\R^{6+z}}\frac{1}{(p^2+m^2)} d^{6+z}p \;. \ees The evaluation of
Feynman graphs of a Lagrangian under $\zeta$-function evaluation and
dimensional regularization both yield results in
$\C\{\{z\}\}$. Therefore they can both be expressed as sections of
this renormalization bundle. These sections both satisfy condition
\eqref{local}. After a suitable change of coordinates,
$\zeta$-function regularization is a Mellin transform of dimensional
regularization \bes \gamma_\zeta(s) = \int_0^\infty
z^s\gamma_{\rm{dim}}(z)\frac{dz}{z} \;. \ees Therefore, the
evaluations of a $\zeta$-function regularization and dimensional
regularization on Feynman diagrams differs only by multiplication by a
holomorphic function in $z$. Therefore, one expects \bes
(\gamma_{\textrm{dim}})_t^* \omega \sim (\gamma_\zeta)^*_t\omega \;,
\ees and that the two regularization schemes to give the same
$\beta$-function. This is a well known fact, established in many
papers, including \cite{Speer}.
\end{example}

The global connection, $\omega$, gives a way of translating between
different regularization schemes that are not gauge equivalent. That
is, it defines a relationship between two different $\beta$-functions,
or any value of $\tilde{R}$ by gauge transformations. Unrenormalizable
regularization schemes, i.e. those that do not satisfy condition
\eqref{local} are not well studied or understood because of their
unphysical nature.  The fact that they can be related in any way to
renormalizable regularization schemes is very surprising. The
implications that this may have about the structure of perturbatively
solving Quantum Field Theories is left for future work.

A second implication of the global connection is that two different
Lagrangians that share the same Hopf algebra are now sections that
differ only by a gauge transformation.

\begin{example}
The Lagrangians for a QFT in different gravitational settings have
different parameters, as some parts of the Lagrangian depends on the
background curvature. Explicitly, if $a^\dagger(\vec{p})$ and
$a(\vec{p})$ are the raising and lowering operators for a Fock space,
and $g$ is the metric of the space time, the field $\phi(\vec{x})$ can
locally be expressed in terms of the metric as \bes \int_{\R^6}
\frac{1}{(2\pi)^6\sqrt{2}(p^jg^i_jp_i +
  m^2)^{\frac{1}{4}}}a^\dagger(\vec{p})e^{i p^jg^i_jx_i}+
a(\vec{p})e^{-ip^jg^i_jx_i} \sqrt{|\det(g)|}d^6p \; .\ees Instead of
$|d\phi|^2$, the Laplacian operator is $\Delta(g) =
\frac{1}{\sqrt{|g|}} \partial_i (\sqrt{|g|}g^{ij}\partial_j)$. The
corresponding Lagrangian is now a function of the curvature, \bes
\mathcal{L} = -\frac{1}{2}\phi(g,x)(\Delta_g + m^2)\phi(g,x) +
\lambda\phi^3(g, x)\; .\ees The associated $\beta$-function will also
be a function of $g$. For a fixed regularization scheme, there is a
family of sections of the renormalization bundle $\gamma_g(z)$
parameterized by the curvature. If the regularization scheme is
renormalizable, then $\beta(\gamma_g)$ is well defined. The gauge
transformations on the bundle provide a way of relating the
$\beta$-functions in this family. This shows the existence of a global
$\beta$-function for a scalar field theory as a function of
curvature. Because of the Birkhoff decomposition of the sections
$\gamma_g(z)$ represent decomposition over coordinate patches, this
also shows that BPHZ renormalization is consistent over a curved
space-time manifold, and suggests that this renormalization bundle can
be constructed as a bundle over a curved background space as opposed
to the flat example constructed by Connes and Marcolli in
\cite{CMbook}. Working over a manifold in general, the global
$\beta$-function can be found via $\zeta$-function regularization, but
its construction is beyond the scope of this paper. For more on the
renormalization bundle over curved space time, see \cite{thesis2}.
\end{example}

\bibliographystyle{amsplain}
\bibliography{/home/mithu/bibliography/Bibliography}{}

\end{document}